\documentclass[onecolumn,12pt]{article}

\usepackage{multirow,url}
\usepackage[pdftex]{graphicx}
\usepackage{amsfonts,amsmath,amssymb,amsthm,epsfig,subfigure}
\usepackage{epsfig}

\newcommand{\be}{\begin{eqnarray}}
\newcommand{\ee}{\end{eqnarray}}

\def\eg{{\em e.g.},~}
\def\ie{{\em i.e.},~}

\def\tD{\tilde{D}}
\def\cf{{\em cf.},~}

\newcommand{\ben}{\begin{enumerate}}
\newcommand{\een}{\end{enumerate}}
\newcommand{\beq}{\begin{equation}}
\newcommand{\eeq}{\end{equation}}
\newcommand{\beqa}{\begin{eqnarray*}}
\newcommand{\eeqa}{\end{eqnarray*}}
\newcommand{\bit}{\begin{itemize}}
\newcommand{\eit}{\end{itemize}}
\newcommand{\bt}{\begin{tabular}{c}}
\newcommand{\btt}{\begin{tabular}}
\newcommand{\et}{\end{tabular}}

\newtheorem{corollary}{Corollary}
\newtheorem{lemma}{Lemma}
\newtheorem{theorem}{Theorem}

\begin{document}

\title{Variation of: \\
``The effect of caching on a model of\\
 content and access provider revenues in\\
information-centric networks"\thanks{This
research was supported by NSF CNS grant 1116626.}}

\author{
G. Kesidis\\
CS\&E and EE Depts \\
 Pennsylvania State University \\
 University Park, PA, USA \\
gik2@psu.edu
}

\maketitle

\begin{abstract}
This is a variation of the two-sided market
model of \cite{EconCom13}: Demand 
$D$ is concave in $\tilde{D}$ in (16) of \cite{EconCom13}.
So, in (5) of \cite{EconCom13} and after Theorem 2, 
take the parametric case $0 < a \leq 1$.
Thus, demand $D$ is both decreasing and concave in price $p$, 
and so the utilities ($U=pD$) are also concave in price.
Also, herein 
a simpler illustrative demand-response model is used in
Appendix A and B.
\end{abstract}

\section{Introduction}

In this paper, we consider a game between an 
Internet Service (access) Provider (ISP)
and content provider (CP) on a platform of end-user demand.
A price-concave demand-response is motivated based on 
the delay-sensitive applications that are expected to
be subjected to the assumed usage-priced priority
service over best-effort service.
Thus, we are considering
a two-sided market with multiclass demand wherein 
one class (that under consideration herein) is delay-sensitive.
Both the Internet and proposed Information Centric Network
(ICN, encompassing Content Centric Networking (CCN)) 
scenarios are considered. For our purposes,
the  ICN case is
basically different in the polarity of the side-payment (from
ISP to CP in an ICN) and, more importantly here,
in that content caching by the ISP is
incented. 

Pricing congestible commodities have been extensively
studied. For example,  in \cite{Johari10} a demand model is
is based on a ``cost" that is the sum of
a price and latency term. We herein take this relationship
to be an implicit one in which the latency factor is
also an increasing function of demand. The resulting
price-{\em concave} demand-response model is extended
to account for content caching. The corresponding
Nash equilibria are derived as a function of 
the caching factor.

\section{Problem Set-Up: The Internet model}\label{problem-setup}

Suppose there are two providers, 
one content (CP indexed 2) and the other access (ISP
indexed 1), with 
{\em common} consumer demand-response 
\cite{Economides08}\footnote{Leader-follower dynamics, rather than 
simultaneous play at the same time-scale, are considered in \cite{walrand09}. 
For the problem
setting of this paper, leader-follower dynamics
were considered by us in \cite{TelSys11} and  provider competition in \cite{ReArch10,KKF13}.}. 
First suppose that the demand response to price is linear:
\be\label{linear-DR}
D & = & D_{\max}-d (p_1 + p_2),
\ee
where $d$ is demand sensitivity to the price, $p_1$ and $p_2$ are,
respectively, the prices charged by the ISP and CP,
and $D_{\max}>0$ is the demand at zero usage based price\footnote{Note
that ISPs  are continuing to depart from pure flat-rate pricing
(based on maximum access bandwidth) for unlimited monthly volume,
\eg \cite{Shin06,ATT11}.}. 
Suppose the revenue of the ISP is
\be\label{content-based-rev1}
U_1 & = &  (p_1 + p_s)D,
\ee
where $p_s$ is the side payment from content to access provider.
Similarly, the revenue of the CP is
\be\label{content-based-rev2}
U_2 & = &  (p_2 - p_s)D.
\ee
Consider a noncooperative game played by the CP and ISP adjusting their
prices, respectively $p_2$ and $p_1$, to maximize
their respective revenues, with all
other parameters fixed. In particular, the fixed side-payment $p_s$ 
is here assumed regulated. Note that the utilities are linear
functions of $p_s$ so that if $p_s$ were under the control of one
of the players, $p_s$ would simply be set at an extremal value.

\begin{figure}
\includegraphics[width=2.75in]{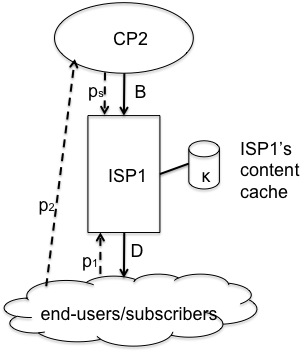}
\caption{ISP and CP game on a platform of end-user
demand-response}\label{ISP_vs_CP}
\end{figure}

The following simple result was shown in \cite{TelSys11,ReArch10}.

\begin{theorem}\label{old-thm}
The interior Nash equilibrium\footnote{In this paper, we do not 
consider boundary Nash equilibria, where at least one player is 
selecting an extremal value for one of their control parameters,
often resulting in that player essentially opting out of the game,
or maximally profiting from it at the expense of the other player.
The boundary equilibria are also specified in \cite{TelSys11}.}
is
\beqa
p_1^* = \frac{D_{\max}}{3d}-p_s &\mbox{and} & 
p_2^* = \frac{D_{\max}}{3d}+p_s
\eeqa
when
\be\label{p_s-assumption}
| p_s| & <&  \frac{D_{\max}}{3d},
\ee
with player utilities
\beqa
U_1^*,U_2^* & = & \frac{D_{\max}^2}{9d}.
\eeqa
\end{theorem}

Note that this result allows
$p_s<0$, \ie net side payment
is from ISP to CP (remuneration for content instead of access
bandwidth). 
But in the Internet setting, we take $p_s>0$, whether
there is direct side-payment from CP to ISP (or, again, indirectly by
payment through  the peering contract 
between the residential ISP
and the ISP of the CP - a contract that penalizes for asymmetric
traffic exchange neutrally based on aggregate traffic volume).

In \cite{ReArch10,kesidis12}, we showed that 
the ISP may actually experience a reduction in revenue/utility
with the introduction of side payments, using
a communal demand model that had different demand-sensitivity-to-price
parameters $d$ per provider type and also multiple providers of each type
(\ie provider competition).
Such a model was also considered in \cite{IFIP11}.

Consider a concave demand response to price, \eg
\be
D & = & 
\left(\frac{1}{D_{\max}-(p_1 + p_2)d} + a
\right)^{-1},
~~a\geq 0,
\label{concave-demand}
\ee
where 
\beqa
p_{\max} & = & D_{\max}/d ~~\mbox{when $a=0$.}
\eeqa

The following is a simple extension of Theorem \ref{old-thm}
accommodating (\ref{concave-demand}).

\begin{theorem}\label{concave-thm}
For utilities (\ref{content-based-rev1}) and 
(\ref{content-based-rev2}),
the interior Nash equilibrium for a strictly concave demand response $D$ is
\be\label{pstar-k}
p^*_1 + p_s ~=~   p^*/2 ~=~ p^*_2 -p_s,
\ee
where $p^*=p^*_1+p^*_2$ solves 
\be\label{pstar-cond}
D(p^*)+D'(p^*) p^* /2 & = & 0.
\ee
and $|p_s|<p^*/2$.
\end{theorem}

For the example of (\ref{concave-demand}) with $a\geq 0$,
\be
p^* &  = & \frac{4aD_{\max}+3-\sqrt{8aD_{\max}+9}}{4ad}
\label{concave-pstar}
\ee
Note that  simply
by L'Hopital's rule, 
$\lim_{a\rightarrow 0}p^* = \tfrac{2}{3}D_{\max}/d = 
\tfrac{2}{3}p_{\max}$,
which is
consistent with Theorem \ref{old-thm}.
Again, under communal demand response with only one provider
of each type,
neither $p^*=p_1^* + p_2^*$ nor $U_1^*$ depend on the side payment $p_s$.
In an illustrative example of Appendix A, the parameter 
\beqa
a & = & (B-\lambda)^{-1},
\eeqa
where $B$ is a possibly congested bandwidth resource,
parameter $0\leq \lambda< B$,  and demand
$D\leq \min\{B-\lambda,D_{\max}\}$.

\section{ICN model}\label{ICN-sec}

Again, in an ICN, residential users request content 
(or, more generally, information regarding application services) 
of the ISP/resolver, 
and the ISP/resolver
decides the content provider.
Therefore in an ICN, it's reasonable to
assume that the side-payment is from ISP to CP,
\ie $p_s <0$.
Also, the ISP is motivated to cache content, unlike 
for our simple Internet case, to reduce the
side payment (\ie avoid paying for, \eg the networking
costs of the ISP-selected CP to transmit the user-requested  content).
Suppose that the ISP decides to cache a fraction $\kappa$ of
the content and this results in lower delay between the CP
and ISP, and a lower required side-payment to the CP,
\cf (\ref{ICN-revenue}). 
If we model mean delay as $1/(B-D)$ \cite{Wolff89}, where $B$ is
the service capacity between CP and ISP, then with
caching factor $\kappa$, this delay is reduced to
$1/(B-(1-\kappa)D)$. For the models of Appendix B,
the demand response:
\begin{itemize}
\item is increasing in caching factor $\kappa$,
\item is concave in price for $\kappa\in[0,1)$, and
\item tends to linear in price (\ref{linear-DR}) as   $\kappa\rightarrow 1$.
\end{itemize}
In an illustrative example of Appendix B, the demand parameter
in (\ref{concave-demand}) is
\beqa
a  & = & (1-\kappa)(B-\lambda)^{-1}.
\eeqa
Note that  neither $D_{\max}$ nor $p_{\max}$ are  assumed
dependent on $\kappa$.
Because of ISP caching, the ISP and CP utilities generalize to
\be
U_1 & = &  (p_1 + (1-\kappa)p_s)D - c(\kappa), \label{ICN-revenue}\\
U_2 & = &  (p_2 - (1-\kappa)p_s)D, \nonumber
\ee
again with $p_s <0$, where
$c(\kappa)$ is the cost of caching borne by the ISP.
In Appendix C, we argue that $c$ is convex in $\kappa$.

Note that
the caching cost $c$ component of $U_1$ does not depend on
$p_2$ or $p_1$, and
$|p_s|<p^*/2$  implies
  $|(1-\kappa)p_s|<p^*/2$.
So, we can 
use the results of Theorem \ref{concave-thm} and (\ref{concave-pstar}) 
here,
with parameters $(1-\kappa)p_s$  and
$a =  (1-\kappa)(B-\lambda)^{-1}$
instead of $p_s$ and 
$a=(B-\lambda)^{-1}$, respectively, 
to obtain the utilities $U_1^*,U_2^*$ at
Nash equilibrium, $p_1^*,p_2^*$.
We can then consider how $U_1^*,U_2^*$ depend 
on the caching factor $\kappa$.


\begin{thebibliography}{10}

\bibitem{TelSys11}
E.~Altman, P.~Bernhard, S.~Caron, G.~Kesidis, J.~Rojas-Mora, and S.~Wong.
\newblock A study of non-neutral networks under usage-based pricing.
\newblock {\em Telecommunication Systems Journal Special Issue on
  Socio-economic Issues of Next Generation Networks}, 2011.

\bibitem{IFIP11}
E.~Altman, A.~Legout, and Y.~Xu.
\newblock Network non-neutrality debate: {An} economic analysis.
\newblock In {\em Proc. IFIP Networking}, 2011.

\bibitem{ATT11}
K.~Bode.
\newblock {AT\&T} to impose caps, overages.
\newblock
  http://www.dslreports.com/shownews/Exclusive-ATT-To-Impose-Caps-Overages-113149,
  Mar. 13, 2011.

\bibitem{ReArch10}
S.~Caron, G.~Kesidis, and E.~Altman.
\newblock Application neutrality and a paradox of side payments.
\newblock In {\em Proc. ACM Re-Architecting the Internet (ReArch) Workshop,
  Philadelphia}, Nov. 2010.

\bibitem{DC10}
G.~Dan and N.~Carlsson.
\newblock Power-law revisited: {A} large scale measurement study of {P2P}
  content popularity.
\newblock In {\em Proc. IPTPS}, 2010.

\bibitem{Economides08}
N.~Economides.
\newblock Net neutrality: Non-discrimination and digital distribution of
  content through the {Internet}.
\newblock {\em I/S: A Journal of Law and Policy}, 4:209--233, 2008.

\bibitem{Johari10}
R.~Johari, G.Y. Weintraub, and B.~Van Roy.
\newblock Investment and market structure in industries with congestion.
\newblock {\em Operations Research}, 58(5), 2010.

\bibitem{kesidis12}
G.~Kesidis.
\newblock Side-payment profitability under convex demand-response modeling
  congestion-sensitive applications.
\newblock In {\em Proc. IEEE ICC}, Ottawa, Canada, June 2012.

\bibitem{KKF13}
F.~Kocak, G.~Kesidis, and S.~Fdida.
\newblock Network neutrality with content caching and its effect on access
  pricing.
\newblock In {\em Smart Data Pricing}. S. Sen and M. Chiang (Eds.), Wiley,
  2013.

\bibitem{EconCom13}
F.~Kocak, G.~Kesidis, T.-M. Pham, and S.~Fdida.
\newblock The effect of caching on a model of content and access provider
  revenues in information-centric networks.
\newblock In {\em Proc. ASE/IEEE EconCom}, Washington, DC, Sept. 2013.

\bibitem{walrand09}
J.~Musacchio, G.~Schwartz, and J.~Walrand.
\newblock A two-sided market analysis of provider investment incentives with an
  application to the net-neutrality issue.
\newblock {\em Review of Network Economics}, 8(1), 2009.

\bibitem{Shin06}
A.~Shin.
\newblock Who's the bandwidth bandit?
\newblock
  http://blog.washingtonpost.com/thecheckout/2006/10/bandwidth\_bandit.html,
  Oct. 4, 2006.

\bibitem{Wolff89}
R.W. Wolff.
\newblock {\em Stochastic Modeling and the Theory of Queues}.
\newblock Prentice-Hall, Englewood Cliffs, NJ, 1989.

\end{thebibliography}

\section*{Appendix A: Explanation of concave demand response}

Consider a price-concave demand response $\tD$. In particular,
for price  $p\in[0,p_{\max}]$, consider the linear case
let 
\be\label{tD-def}
\tD(p)~:=  D_{\max} - pd ~=D_{\max}(1-p/p_{\max}) ~\geq 0.
\ee
Suppose  that the demand $D$ satisfies
\beqa
D  & = & [g(D)\tD]^+ ,
\eeqa
where $g$ is a decreasing and concave factor, not dependent
of price, accounting
for demand loss due to congestion.
For example \cite{EconCom13}, 
\beqa
g(D) & = &\frac{1-\lambda/(B-D)}{1-\lambda/B},
\eeqa
where the term $1/(B-D)$ is taken from the queueing delay of
an M/M/1 queue with mean arrival rate $D$ and
mean service rate $B>D$ \cite{Wolff89}. 
More simply, we can take 
\be
g(D) & = & \frac{B-\lambda-D}{B-\lambda}, \label{linear-g}
\ee
where $B$ is the available bandwidth resource, and
parameter $0\leq \lambda < B$. 
Note that for these examples: 
\begin{itemize}
\item $g(0)=1$, 
\item $g(B-\lambda)=0$, so that
\item 
$0\leq D\leq \min\{B-\lambda,D_{\max}\}$, and
\item $g$ is non-negative, decreasing and concave.
\end{itemize}

\begin{lemma}\label{concave-D-lemma}
If $D   =  g(D)\tD >0 $ with $g$ 
nonnegative, decreasing and concave and $\tD \geq 0$,
then $D$ is increasing and concave in $\tD$.
\end{lemma}

\begin{proof}
Let $g' = \mbox{d}g/\mbox{d}D$ and
$D' = \mbox{d}D/\mbox{d}\tD$.
By direct differentiation with respect to $\tD$:
\beqa
D'  & = & g + \tD g'D' \\
\Rightarrow ~~D'  & = &  \frac{g}{1- \tD g'} ~\geq 0\\
\Rightarrow ~~D''  & = &  \frac{(1- \tD g')g' + (\tD g'' + g')g}
{(1- \tD g')^2} D' ~~\leq 0,
\eeqa
\end{proof}

\begin{corollary}
Under Lemma \ref{concave-D-lemma}  and
if $\tD$ is decreasing and concave in price $p$, then
$D$ is non-negative and decreasing in $p$, and both $D$ and $pD$ 
are concave in $p$.
\end{corollary}

\begin{proof}
By the above lemma,  $D(\tD)\geq 0$ is increasing and concave.
Again, by direct differentiation:
\beqa
\frac{\partial D(\tD)}{\partial p}  & = & D'(\tD) 
\frac{\partial \tD}{\partial p}  ~\leq 0 \\
\frac{\partial^2 D(\tD)}{\partial p^2}  & = & D''(\tD) 
\left(\frac{\partial \tD}{\partial p}\right)^2 
+ D'(\tD)  
\frac{\partial^2 \tD}{\partial p^2} 
~\leq 0 
\eeqa
\end{proof}

For linear demand-response to price
 (\ref{tD-def}) and the linear congestion factor (\ref{linear-g}), 
\be
D(p) & = & (\tD(p)^{-1} + (B-\lambda)^{-1})^{-1} 
\label{demand-response0}\\
 & = & (B-\lambda) \left(1-\frac{1}{1+(D_{\max}-dp)/(B-\lambda)}  \right)
\nonumber
\ee
which is decreasing and concave in $p$, with
\beqa
D(0) & = & (D_{\max}^{-1} + (B-\lambda)^{-1})^{-1}
~\leq~ \min \{ B-\lambda,~D_{\max}\}.
\eeqa
It's also easy to see that 
\beqa
\lim_{B\rightarrow \infty, ~p\rightarrow 0}
D & = &  D_{\max}.
\eeqa

\section*{Appendix B: Explanation of demand increasing
in caching factor}

As a result of ISP caching, only a fraction $(1-\kappa)$ of the demand
$D$ is transmitted through the the bandwidth $B$ between ISP and CP.
So, the congestion factor (\ref{linear-g}) is modified to
\beqa
g_{\kappa}(D) & = & \frac{B-\lambda-(1-\kappa)D}{B-\lambda}\\
& = & \frac{(B-\lambda)/(1-\kappa) - D}{(B-\lambda)/(1-\kappa)}\\
\eeqa
So, solving 
$D=\tD g_\kappa (D) $
results in (\ref{demand-response0}) with
$B-\lambda$ replaced by
$(B-\lambda)/(1-\kappa)$:
\be\label{demand-response1}
D(p) & = & (\tD(p)^{-1} + (1-\kappa)(B-\lambda)^{-1})^{-1}.
\ee
Thus, if positive $\kappa<1$, the demand 
is concave in price $p$ and increasing in $\kappa$.
On the other hand, as  $\kappa\rightarrow 1$, the demand
tends to linear  in price (\ref{linear-DR}).

\begin{lemma}
Generally, if the congestion factor
$g$ is a decreasing function, then the demand
$D$ increases with caching factor $\kappa$. 
\end{lemma}

\begin{proof}
First note  that $g_\kappa(D)=g_0((1-\kappa)D):=g((1-\kappa)D)$, 
is decreasing in $(1-\kappa)D $
(hence increasing in caching factor $\kappa$). 
Consider the solution
\be
D_{\kappa} & = & \tD g_{\kappa}(D_{\kappa}) \label{D_k-defn}
\ee
and note that $D_{\kappa}\geq D_0$. Now,
\beqa
D_0 ~= ~\tD g_0(D_0) & < & \tD g_0((1-\kappa)D_0) ~=~\tD g_{\kappa}(D_0).
\eeqa
So, if $D_{\kappa}\leq D_0$, then we would have
\beqa
D_{\kappa} ~\leq~ D_0 & < & \tD g_{\kappa}(D_0)
~\leq ~ \tD g_{\kappa}(D_{\kappa}),
\eeqa
which contradicts the definition of $D_{\kappa}$ in 
(\ref{D_k-defn}).
\end{proof}

\section*{Appendix C: Convexity of cost of caching as a function of
caching factor}

Assume that the cost of caching is proportional to the number 
of cached items (content),
in turn proportional to the (mean) amount of memory required to store them.
For a fixed population of $N$ end-users (a proximal group served by
an ISP), let $\pi(j)$ be the proportion of the items that
will soon be of interest to precisely $j$ end-users.
Finally, suppose the ISP naturally prioritizes its cache to hold
the most popular content.
So, a ``caching factor" $\kappa$, based on all-or-none decisions to cache
content of the same popularity, would satisfy
\beqa
\kappa & \propto & \sum_{j=N-f(\kappa)}^N j\pi(j).
\eeqa
for some $f(\kappa)\in\{0,1,2,...,N\}$.
The cost of caching would be proportional to the number of cached items, \ie
\beqa
c(\kappa) & \propto & \sum_{j=N-f(\kappa)}^N \pi(j).
\eeqa

Suppose that the great majority of potentially desired content is only minimally popular,
\ie $\pi(j)$ is 
decreasing\footnote{Note that this general assumption 
obviously accommodates
the empirically observed
 Zipf distribution for content popularity, \eg
\cite{DC10}.}
We now argue that the caching 
cost $c(\kappa)$ is convex  and increasing 
for the simplified continuous scenario 
ignoring the (positive) constants of proportionality:
\beqa
\kappa ~ = ~ \int_{N-f(\kappa)}^N z\pi(z)\mbox{d}z 
&\mbox{and}&
c(\kappa) ~ = ~ \int_{N-f(\kappa)}^N \pi(z)\mbox{d}z,
\eeqa
with $c(0)=0$ and $c(1)=1$.
By differentiating successively, we get
\be
1 & = &  (N-f(\kappa))\pi(N-f(\kappa))f'(\kappa) \label{f-equ}\\
c'(\kappa) & = & \pi(N-f(\kappa))f'(\kappa)\nonumber \\
\Rightarrow 1 & = &  (N-f(\kappa))c'(\kappa)\nonumber \\
\Rightarrow c''(\kappa) & = &  f'(\kappa)(N-f(\kappa))^{-2} \label{c-equ}
\ee
Note that $f'>0$ by (\ref{f-equ}) and therefore
$c''>0$ by (\ref{c-equ}).


\end{document}